\documentclass[12pt]{amsart}

\usepackage{amssymb,latexsym}
\usepackage{amsthm,amsfonts}
\usepackage[centertags]{amsmath}
\usepackage[a4paper,margin=22mm]{geometry}

\newcommand{\N}{{\mathbb{N}}}

\newcommand{\F}{{\mathbb{F}}}
\newcommand{\bP}{{\mathbb{P}}}
\newcommand{\bL}{{\mathcal{L}}}
\newcommand{\cG}{{\mathcal{G}}}

\newcommand{\s}{{\bf s}}

\def \se {sequence }
\def \ses {sequences }
\def \com {complexity }

\newtheorem{theorem}{Theorem}

\newtheorem{corollary}{Corollary}
\newtheorem{lemma}{Lemma}

\newtheorem{definition}{Definition}

\newtheorem{remark}{Remark} 

\begin{document}
\title{Sequences with high nonlinear complexity}

\author{HARALD NIEDERREITER and CHAOPING XING}
\address{Harald Niederreiter,
Johann Radon Institute for Computational and Applied Mathematics,
Austrian Academy of Sciences,
Altenbergerstr. 69,
A-4040 Linz, AUSTRIA,
and Department of Mathematics, University of Salzburg, 
Hellbrunnerstr. 34, A-5020 Salzburg, AUSTRIA;
Chaoping Xing,
School of Physical and Mathematical Sciences, 
Nanyang Technological University,
Singapore 637371, REPUBLIC OF SINGAPORE
}
\email{ghnied@gmail.com (H. Niederreiter), xingcp@ntu.edu.sg (Chaoping Xing)}



\date{\today} 



\begin{abstract}
We improve lower bounds on the $k$th-order nonlinear complexity of pseudorandom sequences over finite fields and we establish a probabilistic result
on the behavior of the $k$th-order nonlinear complexity of random sequences over finite fields.
\end{abstract}

\maketitle

\section{Introduction} \label{sein}

Pseudorandom \ses over large finite fields are of interest for simulation methods since such \ses can be transformed easily into \ses of uniform
pseudorandom numbers in the unit interval $[0,1]$ (see \cite[Chapter~8]{N92}). Another area of applications is cryptography. In order to assess the
suitability of a pseudorandom sequence, complexity-theoretic and statistical requirements have to be tested. In practice, both categories of 
tests---complexity-theoretic and statistical---should be carried out since these two categories are in a sense independent (see e.g. the recent 
paper~\cite{N12}). A classical survey article on the testing of pseudorandom \ses in a cryptographic context is~\cite{R92}.

In this paper we focus on the complexity-theoretic analysis of pseudorandom \ses over finite fields.
A variety of \com measures for such \ses is available in the literature. The most common approach is to measure \com by the shortest length
of a feedback shift register that can generate the given sequence. The basic concept of this type is the \emph{linear complexity} (also called the \emph{linear span})
where only linear feedback shift registers are considered (see also Remark~\ref{relk} below). There is a considerable amount of literature on the linear \com
which is surveyed in~\cite{N03}, \cite{TW}, \cite{W10}, and the recent handbook article~\cite{MW13}. Far less work has been done on \com measures referring to feedback
shift registers with feedback functions of higher algebraic degree (we may call them ``nonlinear complexities"). A \com measure of this type which has received
some attention is the \emph{maximum-order complexity} due to Jansen~\cite{Ja}, \cite{Ja91} (see Remark~\ref{remo} below). There are also \com measures for \ses
based on pattern counting, such as the \emph{Lempel-Ziv complexity} (see~\cite{LZ} for the definition and~\cite{M91} for cryptographic applications). The well-known
\emph{Kolmogorov complexity} is not of practical relevance since it cannot be computed in general for \ses of large length.

This paper contributes to the theory of nonlinear complexities by improving lower bounds on nonlinear complexities of interesting pseudorandom \ses 
and by establishing a probabilistic result on the behavior of nonlinear complexities of random sequences. 
In Section~\ref{sede} we collect the basic definitions. In Sections~\ref{seco} and~\ref{sehf} we establish complexity bounds for certain explicit inversive sequences
and for newly constructed \ses from Hermitian function fields, respectively. Finally, in Section~\ref{sepr} we present the mentioned probabilistic result.

\section{Definitions} \label{sede}

We write $\F_q$ for the finite field with $q$ elements, where $q$ is an arbitrary prime power. For any positive integer $m$, let $\F_q[x_1,\ldots,x_m]$ 
be the ring of polynomials over $\F_q$ in the $m$ variables $x_1,\ldots,x_m$. Furthermore, we denote the set of positive integers by $\N$. Now we define
nonlinear complexities for \ses of finite length over $\F_q$.

\begin{definition} \label{decf} {\rm
Let $\s =(s_i)_{i=1}^n$ be a \se of length $n \ge 1$ over the finite field $\F_q$ and let $k \in \N$. If $s_i=0$ for $1 \le i \le n$, then we define
the \emph{$k$th-order nonlinear \com} $N^{(k)}(\s)$ to be $0$. Otherwise, let $N^{(k)}(\s)$ be the smallest $m \in \N$ for which there exists a
polynomial $f \in \F_q[x_1,\ldots,x_m]$ of degree at most $k$ in each variable such that
\begin{equation} \label{eqre}
s_{i+m}=f(s_i,s_{i+1},\ldots,s_{i+m-1}) \qquad \mbox{for } 1 \le i \le n-m.
\end{equation}
}
\end{definition}

\begin{remark} \label{rebd} {\rm
For $n=1$ we have $N^{(k)}(\s)=0$ or $1$. For $n \ge 2$ we always have $0 \le N^{(k)}(\s) \le n-1$, where the upper bound holds since~\eqref{eqre} 
is satisfied for $m=n-1$ and $f$ being the constant polynomial $s_n$. Both extreme values $0$ and $n-1$ can occur. This is trivial for $0$ by
Definition~\ref{decf}. Furthermore, if $\s =(s_i)_{i=1}^n$ with $s_i=0$ for $1 \le i \le n-1$ and $s_n=1$, then $N^{(k)}(\s)=n-1$, since the
assumption $N^{(k)}(\s) \le n-2$ easily leads to a contradiction.}
\end{remark}

\begin{remark} \label{remo} {\rm
In Definition~\ref{decf} it suffices to consider $1 \le k \le q-1$. This follows from the well-known fact that, as a map, any polynomial $f: \F_q^m
\to \F_q$ can be represented by a polynomial over $\F_q$ in $m$ variables of degree at most $q-1$ in each variable (see \cite[pp. 368--369]{LN}). For
$k \ge q-1$ all nonlinear complexities $N^{(k)}(\s)$ of a fixed $\s$ are the same and equal to the \emph{maximum-order complexity} $M(\s)=N^{(q-1)}(\s)$ 
introduced by Jansen~\cite{Ja}, \cite{Ja91}. Connections between the Lempel-Ziv \com and the maximum-order \com were studied in~\cite{Ja}, \cite{LKK06},
\cite{LKK07}.}
\end{remark}  

\begin{remark} \label{relk} {\rm
One may also consider the nonlinear \com $L^{(k)}(\s)$ where in Definition~\ref{decf} we replace ``of degree at most $k$ in each variable" by ``of
total degree at most $k$". It is then trivial that $L^{(k)}(\s) \ge N^{(k)}(\s)$ for any $k$ and $\s$. Note that $L^{(1)}(\s)$ is not quite the same
as the linear complexity $L(\s)$ of $\s$, since in the definition of $L(\s)$ we accept only homogeneous linear polynomials $f \in \F_q[x_1,\ldots,x_m]$
as feedback functions in~\eqref{eqre}, whereas in the definition of $L^{(1)}(\s)$ we accept also linear polynomials with constant term. We have
$L(\s) \ge L^{(1)}(\s) \ge L(\s)-1$ for any $\s$, where the first inequality is trivial and the second inequality follows from a remark in
\cite[p.~401]{LN}. In particular, any lower bound on $L^{(1)}(\s)$, like in Corollaries~\ref{cocf} and~\ref{coci} and in Theorem~\ref{thhf2} below, is also a lower bound
on the linear complexity $L(\s)$. }
\end{remark}

In order to define nonlinear complexities for infinite sequences, we proceed in analogy to the step from the linear \com to the linear \com profile
(see~\cite{MW13}), namely by considering nonlinear complexities of finite-length initial segments of a given infinite sequence.

\begin{definition} \label{deci} {\rm
Let $S=(s_i)_{i=1}^{\infty}$ be an infinite \se over $\F_q$. Then for any $k \in \N$ and $n \in \N$, we define $N_n^{(k)}(S)=N^{(k)}(\s_n)$
and $L_n^{(k)}(S)=L^{(k)}(\s_n)$, where $\s_n =(s_i)_{i=1}^n$. }
\end{definition}

\section{Complexity bounds for explicit inversive sequences} \label{seco}

We first consider \ses of finite length that belong to the family of \emph{explicit inversive pseudorandom sequences} introduced in~\cite{MW04}. 
Let $e$ be a primitive element of $\F_q$, i.e., a generator of the cyclic multiplicative group $\F_q^*$ of nonzero elements of $\F_q$, and choose an
element $a \in \F_q^*$. Let $\s =(s_i)_{i=1}^{q-2}$ be the sequence over $\F_q$ defined by
\begin{equation} \label{eqcf}
s_i=(ae^i-a)^{-1} \qquad \mbox{for } 1 \le i \le q-2.
\end{equation}

\begin{theorem} \label{thcf}
Let $\s =(s_i)_{i=1}^{q-2}$ be the \se over $\F_q$ defined by~\eqref{eqcf}. Then for any integer $k$ with $1 \le k \le q-1$ we have
$$
N^{(k)}(\s_n) \ge (n-1)/(k+1) \qquad \mbox{for } 1 \le n \le q-2,
$$
where $\s_n=(s_i)_{i=1}^n$.
\end{theorem}

\begin{proof}
Since the $k$th-order nonlinear \com is invariant under the termwise multiplication of a \se by an element from $\F_q^*$, we can assume that $a=1$.
The result is trivial for $n=1$, and so we can also assume that $2 \le n \le q-2$. Suppose that $f \in \F_q[x_1,\ldots,x_m]$ with $1 \le m \le n-1$ is a
polynomial of degree at most $k$ in each variable such that
$$
s_{i+m}=f(s_i,s_{i+1},\ldots,s_{i+m-1}) \qquad \mbox{for } 1 \le i \le n-m.
$$
Thus, we have
\begin{equation} \label{eqgp}
-\frac{1}{e^{i+m}-1} + f \left(\frac{1}{e^i-1},\frac{1}{e^{i+1}-1},\ldots,\frac{1}{e^{i+m-1}-1} \right) =0 \qquad \mbox{for } 1 \le i \le n-m.
\end{equation}
Consider the rational function
\begin{equation} \label{eqra}
R(z)=-\frac{1}{e^mz-1} + f \left(\frac{1}{z-1},\frac{1}{ez-1},\ldots,\frac{1}{e^{m-1}z-1} \right) \in \F_q(z).
\end{equation}
Since $1 \le m < q-1$, we have $e^m \ne e^i$ for $0 \le i \le m-1$. Therefore $e^{-m}$ is not a pole of $f(1/(z-1),1/(ez-1),\ldots,1/(e^{m-1}z-1))$,
and so $R(z) \ne 0 \in \F_q(z)$. Write $R(z)$ in reduced form as $R(z)=v(z)/w(z)$ with $v(z),w(z) \in \F_q[z]$, $v(z) \ne 0$, $w(z) \ne 0$, and $\gcd(v(z),w(z))=1$. 
From~\eqref{eqgp} we get $R(e^i)=0$ for $1 \le i \le n-m$. Therefore $v(z)$ has at least $n-m$ zeros, and so $\deg(v(z)) \ge n-m$. On the other hand,
the definition of $R(z)$ in~\eqref{eqra} implies that $\deg(v(z)) \le \deg(w(z)) \le km+1$, and so $km+1 \ge n-m$. This yields $m \ge (n-1)/(k+1)$,
which is the desired bound.
\end{proof}

\begin{corollary} \label{cocf}
Let $\s =(s_i)_{i=1}^{q-2}$ be the \se over $\F_q$ defined by~\eqref{eqcf}. Then for any integer $k$ with $1 \le k \le q-1$ we have
$$
L^{(k)}(\s_n) \ge (n-1)/(k+1) \qquad \mbox{for } 1 \le n \le q-2,
$$
where $\s_n=(s_i)_{i=1}^n$.
\end{corollary}

\begin{proof}
This follows from Theorem~\ref{thcf} and Remark~\ref{relk}.
\end{proof}

For $k=1$, it follows from Corollary~\ref{cocf} and an inequality in Remark~\ref{relk} that for the linear \com $L(\s_n)$ of the initial segment $\s_n=(s_i)_{i=1}^n$
in Corollary~\ref{cocf} we have $L(\s_n) \ge (n-1)/2$ for $1 \le n \le q-2$. This improves on the lower bound $L(\s_n) \ge (n-1)/3$ shown in \cite[Theorem~1]{MW04}. 

Now we consider infinite periodic \ses belonging to the family of explicit inversive pseudorandom \ses introduced in~\cite{MW04}. Let $d$ be a positive divisor 
of $q-1$ with $d < q-1$ and let $u$ be an element of order $d$ of the 
multiplicative group $\F_q^*$. Such an element can be obtained as $u=e^{(q-1)/d}$, where $e$ is a primitive element of $\F_q$. Furthermore, choose
$b,c \in \F_q^*$ such that $cb^{-1}$ does not belong to the cyclic subgroup of $\F_q^*$ generated by $u$. Then we define the sequence $S=(s_i)_{i=1}^{\infty}$ by  
\begin{equation} \label{eqci}
s_i=(bu^i-c)^{-1} \qquad \mbox{for all } i \ge 1.
\end{equation}
Note that the sequence $S$ is periodic with least period $d$.

\begin{theorem} \label{thci}
Let $S=(s_i)_{i=1}^{\infty}$ be the \se over $\F_q$ defined by~\eqref{eqci}. Then for any integer $k$ with $1 \le k \le q-1$ we have
$$
N_n^{(k)}(S) \ge \min \, \{(n-1)/(k+1),(d-1)/k\} \qquad \mbox{for all } n \ge 1.
$$
\end{theorem}

\begin{proof}
Since the $k$th-order nonlinear \com is invariant under the termwise multiplication of a \se by an element from $\F_q^*$, we can assume that $b=1$ and that $c$ does not
belong to the cyclic subgroup of $\F_q^*$ generated by $u$.
We can also assume that $n \ge 2$ and $N_n^{(k)}(S) < (d-1)/k$, for otherwise the result is trivial. Suppose that $f \in \F_q[x_1,\ldots,x_m]$ with 
$1 \le m \le n-1$ and $m < (d-1)/k$ is a polynomial of degree at most $k$ in each variable such that
$$
s_{i+m}=f(s_i,s_{i+1},\ldots,s_{i+m-1}) \qquad \mbox{for } 1 \le i \le n-m.
$$
Thus, we have 
\begin{equation} \label{eqhp}
-\frac{1}{u^{i+m}-c}+f \left(\frac{1}{u^i-c},\frac{1}{u^{i+1}-c},\ldots,\frac{1}{u^{i+m-1}-c} \right)=0 \qquad \mbox{for } 1 \le i \le n-m.
\end{equation}
Consider the rational function
\begin{equation} \label{eqrt}
R(z)=-\frac{1}{u^mz-c}+f \left(\frac{1}{z-c},\frac{1}{uz-c},\ldots,\frac{1}{u^{m-1}z-c} \right) \in \F_q(z).
\end{equation}
Since $1 \le m < (d-1)/k \le d-1$, we have $u^m \ne u^i$ for $0 \le i \le m-1$. Therefore $cu^{-m}$ is not a pole of $f(1/(z-c),1/(uz-c),\ldots,
1/(u^{m-1}z-c))$, and so $R(z) \ne 0 \in \F_q(z)$. Write $R(z)$ in reduced form as $R(z)=v(z)/w(z)$ with $v(z),w(z) \in \F_q[z]$, $v(z) \ne 0$,
$w(z) \ne 0$, and $\gcd(v(z),w(z))=1$. From~\eqref{eqhp} we get $R(u^i)=0$ for $1 \le i \le n-m$. Therefore $v(z)$ has at least $\min \, \{n-m,d\}$
zeros, and so $\deg(v(z)) \ge \min \, \{n-m,d\}$. On the other hand, the definition of $R(z)$ in~\eqref{eqrt} implies that $\deg(v(z)) \le \deg(w(z))
\le km+1$, and so
$$
km+1 \ge \min \, \{n-m,d\}.
$$
Now $m < (d-1)/k$ yields $km+1 < d$, and so we must have $\min \, \{n-m,d\} =n-m$. Therefore $km+1 \ge n-m$, hence $m \ge (n-1)/(k+1)$, and the proof
is complete.
\end{proof}

\begin{corollary} \label{coci}
Let $S=(s_i)_{i=1}^{\infty}$ be the \se over $\F_q$ defined by~\eqref{eqci}. Then for any integer $k$ with $1 \le k \le q-1$ we have
$$
L_n^{(k)}(S) \ge \min \, \{(n-1)/(k+1),(d-1)/k\} \qquad \mbox{for all } n \ge 1.
$$
\end{corollary}

\begin{proof}
This follows from Theorem~\ref{thci} and Remark~\ref{relk}.
\end{proof}  

The lower bounds on nonlinear complexities in Theorem~\ref{thci} and Corollary~\ref{coci} are better than those for the periodic sequences over $\F_q$
(inversive generators, quadratic exponential generators, general nonlinear generators) shown in~\cite{GSW} and~\cite{MW03}. The exact value of the linear \com of any 
finite-length initial segment of the \se defined by~\eqref{eqci} is known from \cite[Corollary~7]{MW05}. Distribution properties and structural properties of this \se
were investigated in~\cite{W06}.

\section{Sequences obtained from Hermitian function fields} \label{sehf}

The length of the sequence~\eqref{eqcf} over $\F_q$ has order of magnitude $q$ and the period length of the sequence~\eqref{eqci} over $\F_q$ has an order of magnitude at most $q$.
In this section, we construct finite-length \ses over $\F_q$ with high nonlinear \com for which the length has an order of magnitude larger than $q$. This new construction 
of \ses uses the theory of global function fields. We follow the monographs~\cite{NX09} and~\cite{St} with regard to the notation and terminology for global function fields.

Let $F/\F_q$ be a global function field with full constant field $\F_q$. We write $\bP_F$ for the set of places of $F$. Let $\deg(P)$ denote the degree of the place
$P \in \bP_F$. If $\deg(P)=1$, then we speak of a \emph{rational place} of $F$. Let $\nu_P$ be the normalized discrete valuation corresponding to $P \in \bP_F$. For a divisor
$D$ of $F$, let $\bL(D)$ be the Riemann-Roch space associated with $D$. We note that $\bL(D)$ is a finite-dimensional vector space over $\F_q$. Let $\deg(D)$ denote the
degree of the divisor $D$. By the Riemann-Roch theorem \cite[Theorem 1.5.17]{St} we have
\begin{equation} \label{eqrr}
\dim(\bL(D))=\deg(D)+1-g \ \ \ \mbox{whenever } \deg(D) \ge 2g-1,
\end{equation}
where $g$ is the genus of $F$. For $P \in \bP_F$ and $h \in F$ with $\nu_P(h) \ge 0$, we write $h(P)$ for the residue class of $h$ modulo $P$ (see~\cite[p.~6]{St}).
If $P$ is a rational place, then $h(P) \in \F_q$.

Now let $H/\F_q$ be the Hermitian function field over $\F_q$ which exists whenever $q$ is a square, say $q=\ell^2$ with a prime power $\ell$. The Hermitian function
field $H/\F_q$ can be defined explicitly by $H=\F_q(x,y)$ with $y^{\ell}+y=x^{\ell +1}$. The function field $H/\F_q$ has exactly $\ell^3 +1$ rational places and genus
$g=\ell (\ell -1)/2$. A summary of the properties of $H/\F_q$ can be found in \cite[Lemma 6.4.4]{St}. We single out the rational place $P_{\infty} \in \bP_H$ which is
defined as the unique pole of $x$. 

Let $\cG ={\rm Aut}(H/\F_q)$ be the group of field automorphisms of the Hermitian function field $H/\F_q$ that fix the elements of $\F_q$. We refer to
\cite[Section~II]{XD} for a summary of the properties of the group $\cG$. If $\sigma \in \cG$ and $P \in \bP_H$, then the set $\sigma(P) := \{\sigma(h):
h \in P\}$ is again a place of $H$. We have the following simple facts (see \cite[Section~8.2]{St} and \cite[Lemma~2.1]{XD}).

\begin{lemma} \label{leau}
For any $\sigma \in \cG ={\rm Aut}(H/\F_q)$, $P \in \bP_H$, and $h \in H$ we have:\\
{\rm (i)} $\deg(\sigma(P))=\deg(P)$;\\
{\rm (ii)} $\nu_{\sigma(P)}(\sigma(h))=\nu_P(h)$;\\
{\rm (iii)} $\sigma(h)(\sigma(P))=h(P)$ if $\nu_P(h) \ge 0$.
\end{lemma}

Now, using the same notation as in \cite[Lemma~2.2]{XD}, let $\phi$ be the element of $\cG$ determined by
$$
\phi(x)=ex, \ \ \ \phi(y)=e^{\ell +1}y,
$$
where $e$ is a primitive element of $\F_q$. Then according to \cite[Lemma~2.2]{XD}, the rational place $P_{\infty}$ of $H$ satisfies $\phi(P_{\infty})=P_{\infty}$,
and under the action of $\phi$ on $\bP_H$ there are $\ell$ orbits each containing exactly $q-1$ distinct rational places of $H$. We denote these $(q-1) \ell$ distinct rational
places of $H$ occurring altogether in these $\ell$ orbits by 
$$
Q, \phi(Q), \ldots , \phi^{q-2}(Q), P_1, \phi(P_1), \ldots , \phi^{q-2}(P_1), \ldots,P_{\ell -1},\phi(P_{\ell -1}), 
\ldots, \phi^{q-2}(P_{\ell -1}). 
$$

By~\eqref{eqrr} we have $\dim(\bL((2g-1)P_{\infty}+Q))=g+1$ and $\dim(\bL((2g-1)P_{\infty}))=g$, and so we can choose an element $h \in \bL((2g-1)P_{\infty}+Q) \setminus 
\bL((2g-1)P_{\infty})$. Then we consider the \se $\s =(s_i)_{i=1}^M$ over $\F_q$ of length $M :=(q-1)(\ell -1)$ given by
\begin{equation} \label{eqhf}
\s =\left(h(P_1),h(\phi(P_1)),\ldots,h(\phi^{q-2}(P_1)),\ldots,h(P_{\ell -1}),h(\phi(P_{\ell -1})),\ldots,h(\phi^{q-2}(P_{\ell -1})) \right).
\end{equation}
The choice of $h$ guarantees that all terms of the \se $\s$ are well defined. Note that the length $M$ of $\s$ has order of magnitude $q^{3/2}$.

\begin{theorem} \label{thhf1}
Let $H/\F_q$ be the Hermitian function field over $\F_q$ with $q=\ell^2$ for some prime power $\ell$. Let $\s =(s_i)_{i=1}^M$ with $M=(q-1)(\ell -1)$ be the \se over $\F_q$
defined by~\eqref{eqhf}. Then for any integer $k$ with $1 \le k \le q-1$ we have
$$
N^{(k)}(\s_n) \ge \frac{(q-1) \lfloor n/(q-1) \rfloor -1}{\ell (\ell -1)k+ \lfloor n/(q-1) \rfloor} \qquad \mbox{for } 1 \le n \le M,
$$
where $\s_n =(s_i)_{i=1}^n$.
\end{theorem}

\begin{proof}
The result is trivial for $n < q-1$, and so we can assume that $n \ge q-1$. Since $N^{(k)}(\s_n)$ is a nondecreasing function of $n$, we can also assume that $n$ is a
multiple of $q-1$, say $n=(q-1)r$ with $r \in \N$ and $r \le \ell -1$. Now we fix such an $n$. We claim that $\s_n$ is not the zero sequence. For otherwise there exist $n$ 
rational places $Q_1,\ldots,Q_n$ of $H$ different from $P_{\infty}$ and $Q$ that are zeros of $h$. This implies that $h \in \bL(D)$ with
$$
D := (2g-1)P_{\infty} +Q-Q_1 - \cdots - Q_n.
$$
But 
$$
\deg(D)=2g-n=\ell (\ell -1)-n \le \ell (\ell -1) -(q-1)=- \ell +1 < 0,
$$ 
and so $h=0$ by \cite[Corollary 3.4.4]{NX09}. This is a contradiction to the fact that $h \notin \bL((2g-1)P_{\infty})$ by the choice of $h$.

Thus we have $N^{(k)}(\s_n) \ge 1$. If $N^{(k)}(\s_n) \ge q-1$, then the lower bound in the theorem holds trivially. Hence we can assume that $N^{(k)}(\s_n) \le q-2$.
Suppose that $f \in \F_q[x_1,\ldots,x_m]$ with $1 \le m \le q-2 \le n-1$ is a polynomial of degree at most $k$ in each variable such that
\begin{equation} \label{eqge}
s_{i+m}=f(s_i,s_{i+1},\ldots,s_{i+m-1}) \qquad \mbox{for } 1 \le i \le n-m.
\end{equation}
By applying~\eqref{eqge} only for $i=(q-1)(j-1)+t+1$ with $j=1,\ldots,r$ and $t=0,1,\ldots,q-m-2$, we obtain
$$
-h(\phi^{t+m}(P_j))+f \left(h(\phi^t(P_j)),h(\phi^{t+1}(P_j)),\ldots,h(\phi^{t+m-1}(P_j)) \right) =0
$$
for $1 \le j \le r$ and $0 \le t \le q-m-2$. Lemma~\ref{leau}(iii) yields
$$
h(\phi^{t+b}(P_j))=h(\phi^b(\phi^t(P_j)))=\phi^{-b}(h)(\phi^t(P_j))
$$
for $1 \le j \le r$ and all integers $t \ge 0$ and $b \ge 0$, and so
\begin{equation} \label{eqsp}
-\phi^{-m}(h)(\phi^t(P_j))+f \left( h(\phi^t(P_j)),\phi^{-1}(h)(\phi^t(P_j)),\ldots,\phi^{-(m-1)}(h)(\phi^t(P_j)) \right) =0
\end{equation}
for $1 \le j \le r$ and $0 \le t \le q-m-2$.

Consider the element
$$
w=-\phi^{-m}(h)+f \left(h,\phi^{-1}(h),\ldots,\phi^{-(m-1)}(h) \right) \in H.
$$
We have $\nu_Q(h)=-1$ by the choice of $h$, hence $\nu_{\phi^{-m}(Q)}(\phi^{-m}(h))=-1$ by Lemma~\ref{leau}(ii), and so the place $\phi^{-m}(Q)$ is a pole of $\phi^{-m}(h)$.
On the other hand, for $b=0,1,\ldots,m-1$, the place $\phi^{-m}(Q)$ is not a pole of $\phi^{-b}(h)$ (use again Lemma~\ref{leau}(ii) and the choice of $h$), and so
$\phi^{-m}(Q)$ is not a pole of $f \left(h,\phi^{-1}(h),\ldots,\phi^{-(m-1)}(h) \right)$. Hence we must have $w \ne 0$.

Now we study the zeros and poles of $w$. First of all, it follows from~\eqref{eqsp} that all the $(q-m-1)r$ distinct places $\phi^t(P_j)$, $1 \le j \le r$, 
$0 \le t \le q-m-2$, are zeros of $w$. Therefore the degree of the zero divisor $(w)_0$ of $w$ satisfies
$$
\deg((w)_0) \ge (q-m-1)r.
$$
By the choice of $h$ and Lemma~\ref{leau}(ii), the only possible poles of $w$ are the rational places $P_{\infty},Q,\phi(Q),\ldots,\phi^{q-2}(Q)$. Note that $\nu_{P_{\infty}}(h)
\ge -(2g-1)$. Since $P_{\infty}$ is invariant under $\phi$, Lemma~\ref{leau}(ii) shows that $\nu_{P_{\infty}}(\phi^{-b}(h)) \ge -(2g-1)$ for any integer $b \ge 0$. It follows
that $\nu_{P_{\infty}}(w) \ge -(2g-1)km$. Now we determine the possible poles of $w$ in the set $\mathcal{Q}=\{Q,\phi(Q),\ldots,\phi^{q-2}(Q)\}$ of rational places. 
The only pole of $h$ in $\mathcal{Q}$ is $Q$ and its pole order is $1$. Furthermore, for any integer $b$ with $1 \le b \le m \le q-2$, Lemma~\ref{leau}(ii) shows that
the only pole of $\phi^{-b}(h)$ in $\mathcal{Q}$ is $\phi^{-b}(Q)=\phi^{q-1-b}(Q)$ and its pole order is $1$. Altogether, the degree of the pole divisor $(w)_{\infty}$
of $w$ satisfies
$$
\deg((w)_{\infty}) \le (2g-1)km+km+1=2gkm+1.
$$
Now $\deg((w)_{\infty})=\deg((w)_0)$ by a fundamental identity for algebraic function fields (see \cite[Theorem 1.4.11]{St}), and so
$$
2gkm+1 \ge \deg((w)_{\infty}) = \deg((w)_0) \ge (q-m-1)r.
$$
It follows that
$$
m \ge \frac{(q-1)r-1}{2gk+r},
$$
which completes the proof of the theorem (recall that we assumed without loss of generality that $n=(q-1)r$).
\end{proof}

Note that the lower bound on $N^{(k)}(\s_n)$ in Theorem~\ref{thhf1} is of order of magnitude $n/(qk)$. If $n$ is of the maximal order of magnitude $q^{3/2}$, then the
lower bound is of order of magnitude $q^{1/2}/k$. In contrast to Section~\ref{seco}, we can obtain a better lower bound for the nonlinear complexity $L^{(k)}$ (see
Remark~\ref{relk}) of the sequence~\eqref{eqhf} than that implied by Theorem~\ref{thhf1}.

\begin{theorem} \label{thhf2}
Let $H/\F_q$ be the Hermitian function field over $\F_q$ with $q=\ell^2$ for some prime power $\ell$. Let $\s =(s_i)_{i=1}^M$ with $M=(q-1)(\ell -1)$ be the \se over $\F_q$
defined by~\eqref{eqhf}. Then for any integer $k \ge 1$ we have
$$
L^{(k)}(\s_n) \ge \frac{(q-1) \lfloor n/(q-1) \rfloor -(\ell^2 -\ell -1)k-1}{k+ \lfloor n/(q-1) \rfloor } \qquad \mbox{for } 1 \le n \le M,
$$
where $\s_n =(s_i)_{i=1}^n$.
\end{theorem}

\begin{proof}
We proceed exactly as in the proof of Theorem~\ref{thhf1}. The only difference is that now $\nu_{P_{\infty}}(w) \ge -(2g-1)k$ since the polynomial $f$ has \emph{total degree}
at most $k$. Therefore
$$
\deg((w)_{\infty}) \le (2g-1)k+km+1,
$$
and this yields the desired result.
\end{proof}

For small $k$ and for $n$ of a larger order of magnitude than $q$, the lower bound on $L^{(k)}(\s_n)$ in Theorem~\ref{thhf2} is of order of magnitude $q$.

\section{A probabilistic result} \label{sepr}

Let $\mu_q$ be the uniform probability measure on $\F_q$ which assigns the measure $1/q$ to each element of $\F_q$. Let $\F_q^{\infty}$ be the \se
space over $\F_q$ and let $\mu_q^{\infty}$ be the complete product probability measure on $\F_q^{\infty}$ induced by $\mu_q$. We say that a property
of sequences $S \in \F_q^{\infty}$ holds $\mu_q^{\infty}$-\emph{almost everwhere} if it holds for a set of sequences $S$ of $\mu_q^{\infty}$-measure $1$.
We may view such a property as a typical property of a random \se over $\F_q$.

\begin{theorem} \label{thpr}
Let $k$ be an integer with $1 \le k \le q-1$. Then $\mu_q^{\infty}$-almost everywhere we have
$$
\liminf_{n \to \infty} \Big(N_n^{(k)}(S)-\frac{\log n}{\log (k+1)} \Big) \ge 0.
$$
\end{theorem}

\begin{proof}
For $m,n \in \N$ with $m \le n-1$, let $T_n^{(k)}(m)$ be the number of sequences $\s$ of length $n$ over $\F_q$ with $N^{(k)}(\s) \le m$. Each \se
$\s =(s_i)_{i=1}^n$ counted by $T_n^{(k)}(m)$ is (not necessarily uniquely) determined by a polynomial $f \in \F_q[x_1,\ldots,x_m]$ of degree at most
$k$ in each variable and by initial values $s_1,\ldots,s_m$ of the recursion~\eqref{eqre}. Since the number of possibilities for $f$ is $q^{(k+1)^m}$,
we have
\begin{equation} \label{eqrn}
T_n^{(k)}(m) \le q^{(k+1)^m +m} \qquad \mbox{for } 1 \le m \le n-1.
\end{equation}
Now fix $\varepsilon > 0$ and put
$$
b_n=\frac{\log n}{\log (k+1)} -\varepsilon \qquad \mbox{for } n=1,2,\ldots
$$
and
$$
A_n=\{S \in \F_q^{\infty} : N_n^{(k)}(S) \le b_n \} \qquad \mbox{for } n=1,2,\ldots .
$$
Then $1 \le \lfloor b_n \rfloor \le n-1$ for sufficiently large $n$, and so~\eqref{eqrn} yields
$$
\mu_q^{\infty}(A_n) = q^{-n} T_n^{(k)}(\lfloor b_n \rfloor ) \le q^{(k+1)^{b_n} +b_n-n}
$$
for sufficiently large $n$. Now for some $0 < \delta < 1$ we have
$$
(k+1)^{b_n} +b_n-n < n \Big(\frac{1}{(k+1)^{\varepsilon}} + \frac{\log n}{n \log(k+1)} -1 \Big) < - \delta n
$$
for sufficiently large $n$, and so $\sum_{n=1}^{\infty} \mu_q^{\infty}(A_n) < \infty$. Then the Borel-Cantelli lemma (see \cite[Lemma~3.14]{Br} and
\cite[p.~228]{Lo}) shows that the set of all $S \in \F_q^{\infty}$ for which $S \in A_n$ for infinitely many $n$ has $\mu_q^{\infty}$-measure $0$.
In other words, $\mu_q^{\infty}$-almost everywhere we have $S \in A_n$ for at most finitely many $n$. It follows then from the definition of $A_n$ 
that $\mu_q^{\infty}$-almost everywhere we have
$$
N_n^{(k)}(S) > b_n =\frac{\log n}{\log (k+1)} - \varepsilon
$$
for sufficiently large $n$. This means that $\mu_q^{\infty}$-almost everywhere we have 
$$
\liminf_{n \to \infty} \Big(N_n^{(k)}(S) - \frac{\log n}{\log (k+1)} \Big) \ge - \varepsilon.
$$
By applying this for all $\varepsilon =1/r$ with $r \in \N$ and noting that the intersection of countably many sets of $\mu_q^{\infty}$-measure $1$ has
again $\mu_q^{\infty}$-measure $1$, we obtain the result of the theorem.
\end{proof}

\begin{remark} \label{repr} {\rm
For $k=q-1$, Theorem~\ref{thpr} says that $\mu_q^{\infty}$-almost everywhere the maximum-order \com $N_n^{(q-1)}(S)$ (see Remark~\ref{remo}) grows
at least like $(\log n)/(\log q)$ as $n \to \infty$. This is in good accordance with the result of Jansen~\cite{Ja} (see also~\cite{EM} and~\cite{JB}) 
that the expected value of
$N_n^{(q-1)}(S)$ behaves asymptotically like $(\log n)/(\log q)$, up to an absolute constant. On the basis of these results, it may be conjectured
that $\mu_q^{\infty}$-almost everywhere we have
$$
\lim_{n \to \infty} \frac{N_n^{(q-1)}(S)}{\log n} = C_q
$$
for some constant $C_q > 0$ depending only on $q$. A similar behavior may be conjectured for $N_n^{(k)}(S)$ with $1 \le k < q-1$, where $C_q$ is
replaced by a constant $C_{q,k} > 0$ depending only on $q$ and $k$. In view of this heuristic that the expected order of magnitude of $N_n^{(k)}(S)$
for random \ses $S$ is $\log n$, it is clear that the \ses considered in Sections~\ref{seco} and~\ref{sehf} can be said to have high nonlinear complexity. }
\end{remark}

\section*{Acknowledgments}

We are grateful to Arne Winterhof of the Austrian Academy of Sciences for very fruitful discussions on the topic of this paper.
The first author enjoyed the hospitality of Nanyang Technological University in Singapore at the time when this project was initiated.

\end{document}